\documentclass{article}
\usepackage[english]{babel}
\usepackage{amsmath}
\usepackage{amssymb}
\usepackage{amsthm}
\usepackage{geometry}
\usepackage[utf8]{inputenc}
\usepackage[T1]{fontenc}
\usepackage{microtype}

\usepackage{newtxtext,newtxmath}
\usepackage[hidelinks]{hyperref}
\usepackage{booktabs}
\usepackage{graphicx}
\usepackage{enumitem}
\usepackage{bm}
\usepackage[numbers]{natbib}
\usepackage{authblk}

\usepackage{subcaption}

\title{Re-examining the Legendre-Gauss-Lobatto Pseudospectral Methods for Optimal Control}
\author{Yilin Zou\thanks{Ph.D. candidate, School of Aerospace Engineering, zouyl22@mails.tsinghua.edu.cn.} }
\author{Fanghua Jiang\thanks{Associate Professor, School of Aerospace Engineering, jiangfh@tsinghua.edu.cn (Corresponding Author).}}
\affil{Tsinghua University, Beijing, China}
\date{\today}

\newtheorem{prop}{Proposition}

\begin{document}

\maketitle

\begin{abstract} 
    Pseudospectral methods represent an efficient approach for solving optimal control problems. While Legendre-Gauss-Lobatto (LGL) collocation points have traditionally been considered inferior to Legendre-Gauss (LG) and Legendre-Gauss-Radau (LGR) points in terms of convergence properties, this paper presents a rigorous re-examination of LGL-based methods. We introduce an augmented formulation that enhances the standard LGL collocation approach by incorporating an additional degree of freedom (DOF) into the interpolation structure. We demonstrate that this augmented formulation is mathematically equivalent to the integral formulation of the LGL collocation method. Through analytical derivation, we establish that the adjoint system in both the augmented differential and integral formulations corresponds to a Lobatto IIIB discontinuous collocation method for the costate vector, thereby resolving the previously reported convergence issues. Our comparative analysis of LG, LGR, and LGL collocation methods reveals significant advantages of the improved LGL approach in terms of discretized problem dimensionality and symplectic integration properties. Numerical examples validate our theoretical findings, demonstrating that the proposed LGL-based method achieves comparable accuracy to LG and LGR methods while offering superior computational performance for long-horizon optimal control problems due to the preservation of symplecticity.
\end{abstract}

\section{Introduction}
Pseudospectral methods are numerical techniques for solving general nonlinear optimal control problems, with applications in trajectory optimization~\cite{kellyIntroductionTrajectoryOptimization2017} and model predictive control~\cite{wangOptimalRocketLanding2019}. These methods have gained significant attention in the research community due to their computational efficiency and high approximation accuracy when solving complex control problems.

Optimal control problems involve optimizing functions over a continuous time domain, subject to dynamics, path constraints, and boundary conditions. Pseudospectral methods discretize these problems by approximating state and control variables with interpolation polynomials at specifically selected collocation points. By enforcing the dynamics and constraints at these points, the continuous optimal control problem transforms into a nonlinear programming (NLP) problem. This discretized formulation can then be solved using established optimization algorithms designed for large-scale problems with sparse structures, such as IPOPT~\cite{wachterImplementationInteriorpointFilter2006} and SNOPT~\cite{gillSNOPTSQPAlgorithm2002}. Most commonly used pseudospectral software includes GPOPS~\cite{pattersonGPOPSIIMATLABSoftware2014a}, PSOPT~\cite{5612676}, and DIDO~\cite{rossEnhancementsDIDOOptimal2020}.

To achieve high function approximation and integration accuracy, collocation points are selected based on the roots of orthogonal polynomials. Three commonly used sets are the Legendre-Gauss (LG)~\cite{bensonDirectTrajectoryOptimization2006a,gargPseudospectralMethodsSolving2011}, Legendre-Gauss-Radau (LGR)~\cite{gargDirectTrajectoryOptimization2011a,gargPseudospectralMethodsSolving2011}, and Legendre-Gauss-Lobatto (LGL)~\cite{elnagarPseudospectralLegendreMethod1995,fahrooCostateEstimationLegendre2001} points. These sets differ primarily in their inclusion of interval endpoints: LG points exclude both endpoints, LGR points include one endpoint, and LGL points include both endpoints. These sets also offer different orders of polynomial approximation and integration accuracy: LG points can accurately integrate polynomials of degree up to $2N-1$, LGR points up to $2N-2$, and LGL points up to $2N-3$~\cite{gargUnifiedFrameworkNumerical2010a}.

The relationship between the Lagrange multipliers of the discretized NLP problem and the costate variables of the continuous optimal control problem has been comprehensively studied. \citet{rossReviewPseudospectralOptimal2012a} provided an overview of the covector mapping principle, establishing the existence of multipliers that converge to the continuous costate variable. \citet{gargUnifiedFrameworkNumerical2010a} presented the costate approximation scheme for the LG, LGR, and LGL points. The work demonstrated that under the standard LGL pseudospectral method in differential form, the costate fails to converge to the continuous costate equation as the number of collocation points increases. This issue arises from the null space in the discretized costate dynamics equation. \citet{francolinCostateApproximationOptimal2015} further established the theoretical foundations for costate approximation schemes for the LG and LGR points, but similar analysis for the LGL method remained incomplete.

Previous research has established that LGL points in differential form exhibit inferior accuracy and convergence properties compared to LG and LGR points~\cite{gargUnifiedFrameworkNumerical2010a,gargOverviewThreePseudospectral2010}. Although \citet{gongConnectionsCovectorMapping2008} proposed a ``closure conditions'' approach to determine the approximate costate vector, the accurate convergence of the costate for the LGL points is not guaranteed. In this paper, we re-examine the LGL points and introduce an enhanced formulation that addresses the limitations of standard LGL pseudospectral methods by incorporating an additional degree of freedom (DOF) into the interpolation structure. We demonstrate that this augmented formulation is mathematically equivalent to the integral formulation of the LGL collocation method, where the additional DOF remains implicit. We derive the adjoint system for both the augmented differential formulation and the integral formulation, showing that it corresponds to a Lobatto IIIB discontinuous collocation method for the costate vector with respect to the continuous costate equation. This finding validates the covector mapping property of the LGL points. Based on this new formulation, we conduct a comparative analysis of LG, LGR, and LGL points, revealing the advantages of LGL points in terms of discretized problem dimensionality and symplectic properties. We present two numerical examples that demonstrate the effectiveness of the proposed method.

This paper is organized as follows. Section~\ref{sec:problem_formulation} defines the optimal control problem and establishes the first-order necessary conditions for optimality. Section~\ref{sec:augmented_formulation} proposes an augmented polynomial interpolation scheme and derives the corresponding adjoint system in differential form. In Section~\ref{sec:integral_formulation}, we demonstrate the mathematical equivalence between our augmented LGL collocation method and the integral formulation, along with the derivation of the adjoint system in integral form. Section~\ref{sec:comparison} provides a comparative analysis of the LG, LGR, and LGL collocation methods. Section~\ref{sec:examples} presents numerical examples that validate our theoretical findings and illustrate the performance of the proposed method. Finally, Section~\ref{sec:conclusion} summarizes our contributions and offers concluding remarks.

\section{Problem Formulation and First-Order Necessary Conditions for Optimality}\label{sec:problem_formulation}
This work builds upon and extends the formulation presented in~\citet{gargUnifiedFrameworkNumerical2010a} and~\citet{francolinCostateApproximationOptimal2015}. For consistency and clarity, we adopt the notational conventions established in the literature. We present a concise formulation of the optimal control problem below both to maintain self-containment and to establish the necessary mathematical foundation for our subsequent analysis.

We formulate a continuous-time optimal control problem on the time interval $[-1, 1]$. Let $\bm{x}(\tau) \in \mathbb{R}^n$ denote the state vector and $\bm{u}(\tau) \in \mathbb{R}^m$ denote the control vector. The problem is expressed as
\begin{align}
    \min_{\bm{x},\bm{u}} & \quad \Phi\left(\bm{x}(1)\right) + \int_{-1}^1 g(\bm{x}(\tau), \bm{u}(\tau))\,\mathrm{d}\tau\\
    \text{s.t.} & \quad \dot{\bm{x}}(\tau) = \bm{f}(\bm{x}(\tau), \bm{u}(\tau)), \quad \tau \in [-1,1]\\
    & \quad \bm{x}(-1) = \bm{x}_0
\end{align}
where $\Phi: \mathbb{R}^n \rightarrow \mathbb{R}$ represents the terminal cost function, $g: \mathbb{R}^n \times \mathbb{R}^m \rightarrow \mathbb{R}$ is the running cost function, $\bm{f}: \mathbb{R}^n \times \mathbb{R}^m \rightarrow \mathbb{R}^n$ is the system dynamics function, and $\bm{x}_0 \in \mathbb{R}^n$ specifies the initial state condition.

The Hamiltonian associated with the optimal control problem is defined as
\begin{equation}
    \mathcal{H}(\bm{x}, \bm{\lambda}, \bm{u}) = g(\bm{x}, \bm{u}) + \langle\bm{\lambda}, \bm{f}(\bm{x}, \bm{u})\rangle
\end{equation}
where $\bm{\lambda}(\tau) \in \mathbb{R}^n$ represents the costate vector, and $\langle \cdot, \cdot \rangle$ denotes the inner product. 

The first-order necessary conditions for optimality, derived from calculus of variations and Pontryagin's minimum principle (PMP), are given by the following system of equations
\begin{align}
    \bm{\lambda}(1) & = \nabla_{\bm{x}} \Phi\left(\bm{x}(1)\right)\label{eq:costate_boundary}\\
    \dot{\bm{\lambda}}(\tau) & = -\nabla_{\bm{x}} \mathcal{H}(\bm{x}(\tau), \bm{\lambda}(\tau), \bm{u}(\tau)), \quad \tau \in [-1,1]\label{eq:costate_dynamics}\\
    \bm{0} & = \nabla_{\bm{u}} \mathcal{H}(\bm{x}(\tau), \bm{\lambda}(\tau), \bm{u}(\tau)), \quad \tau \in [-1,1]\label{eq:costate_control}
\end{align}

\section{Differential Formulation using LGL Collocation}\label{sec:augmented_formulation}
\subsection{Augmented Polynomial Interpolation Scheme}
The Legendre-Gauss-Lobatto (LGL) nodes comprise a set of $N$ distinct points $\{\tau_i\}_{i=1}^N \subset [-1, 1]$, where $\tau_1 = -1$ and $\tau_N = 1$ are the boundary points, and $\{\tau_i\}_{i=2}^{N-1}$ are the roots of $\dot{P}_{N-1}(\tau)$, the derivative of the Legendre polynomial of degree $N-1$. These nodes satisfy $-1 = \tau_1 < \tau_2 < \cdots < \tau_{N-1} < \tau_N = 1$. Associated with each node $\tau_i$ is a positive quadrature weight $w_i$, forming the LGL quadrature rule. This quadrature rule exactly integrates polynomials of degree up to $2N-3$ over the interval $[-1,1]$~\cite{gargUnifiedFrameworkNumerical2010a}. In this work, we denote the LGL nodes and weights as column vectors $\bm{\tau}, \bm{w} \in \mathbb{R}^N$, respectively.

To address the limitations of standard LGL pseudospectral methods identified in the literature \cite{gargUnifiedFrameworkNumerical2010a,gargOverviewThreePseudospectral2010}, we propose increasing the order of the approximation polynomial for the state variables from degree $N-1$ to $N$. This enhancement is achieved by introducing an additional DOF to the interpolation structure while preserving the essential collocation property. Specifically, we augment the standard Lagrange interpolation polynomial with an additional term as
\begin{equation}
    \bm{x}(\tau) = \sum_{i=1}^{N} \bm{X}_i L_i(\tau) + \bm{X}_p L_p(\tau)\label{eq:interpolation}
\end{equation}
where $\bm{X} \in \mathbb{R}^{N \times n}$ is the matrix of state values at the LGL nodes, with the subscript $i$ denoting the $i$-th row, and $\bm{X}_p \in \mathbb{R}^n$ is a row vector representing the additional coefficient. In \eqref{eq:interpolation}, $L_i(\tau)$ denotes the $i$-th Lagrange interpolating polynomial for the LGL nodes, satisfying $L_i(\tau_j) = \delta_{ij}$, where $\delta_{ij}$ is the Kronecker delta function. The auxiliary $N$-th degree polynomial $L_p(\tau)$ is defined as
\begin{equation}
    L_p(\tau) = \prod_{i=1}^{N}\left(\tau -\tau_i\right)
\end{equation}
By construction, this formulation guarantees that $L_p(\tau_i) = 0$ for all LGL nodes $\{\tau_i\}_{i=1}^N$, thereby ensuring that the values of the approximation polynomial at the collocation points remain exactly equal to $\bm{X}_i$, independent of the newly introduced variable $\bm{X}_p$. Consequently, $\bm{X}_p$ can be viewed as a non-collocation variable that contributes to the interpolation process without representing the polynomial value at any specific point.

The derivative of the augmented polynomial $L_p(\tau)$ evaluated at each LGL node $\tau_i$ constitutes a column vector $\bm{D}_p \in \mathbb{R}^N$. The elements of this vector are given by
\begin{equation}
    \left(\bm{D}_p\right)_i = \dot{L}_p(\tau_i) = \prod_{j=1, j\neq i}^{N}\left(\tau_i -\tau_j\right), \quad i = 1,2,\ldots,N
\end{equation}

To facilitate the analysis of the augmented interpolation scheme, we define the augmented derivative matrix $\tilde{\bm{D}} \in \mathbb{R}^{N \times (N+1)}$ by horizontally concatenating the standard differentiation matrix $\bm{D}$ with the vector $\bm{D}_p$ as
\begin{equation}
    \tilde{\bm{D}} = \begin{bmatrix}
        \bm{D} & \bm{D}_p
    \end{bmatrix}
\end{equation}
Correspondingly, we define the augmented state matrix $\tilde{\bm{X}} \in \mathbb{R}^{(N+1) \times n}$ by vertically concatenating the matrix of nodal state values $\bm{X}$ with the additional coefficient vector $\bm{X}_p$ as
\begin{equation}
    \tilde{\bm{X}} = \begin{bmatrix}
        \bm{X} \\ \bm{X}_p
    \end{bmatrix}
\end{equation}
This augmented formulation allows us to express the state derivative computation concisely while preserving the essential properties of the interpolation scheme.

The discretized NLP problem is given by
\begin{align}
    \min_{\bm{X}, \bm{U}} & \quad \Phi\left(\bm{X}_N\right) + \langle \bm{w}, \bm{G}(\bm{X}, \bm{U})\rangle\label{eq:discretized_cost}\\
    \text{s.t.} & \quad \bm{F}(\bm{X}, \bm{U}) = \bm{D} \bm{X} + \bm{D}_p \bm{X}_p = \tilde{\bm{D}} \tilde{\bm{X}}\label{eq:discretized_dynamics}\\
    & \quad \bm{X}_1 = \bm{x}_0\label{eq:discretized_boundary}
\end{align}
where $\bm{U} \in \mathbb{R}^{N \times m}$ is the matrix of control values at the LGL nodes, $\bm{G} \in \mathbb{R}^N$ is the vector of running costs evaluated at the LGL nodes with elements computed using $\bm{X}$ and $\bm{U}$, and $\bm{F} \in \mathbb{R}^{N \times n}$ is the matrix of the dynamics function evaluated at the LGL nodes.

The following property of the augmented derivative matrix $\tilde{\bm{D}}$ demonstrates how the introduction of additional degrees of freedom resolves the singularity issues inherent in the standard LGL differentiation matrix $\bm{D}$.
\begin{prop}\label{prop:inverse}
    Let $\tilde{\bm{D}}_{2:N+1}$ denote the $N \times N$ submatrix of $\tilde{\bm{D}}$ comprising its columns from 2 to $N+1$, and let $\tilde{\bm{D}}_1$ denote the first column of $\tilde{\bm{D}}$. Then:
    \begin{enumerate}[label=(\roman*)]
        \item The matrix $\tilde{\bm{D}}_{2:N+1}$ is nonsingular.
        \item The following relation holds between the first column and the remaining columns of $\tilde{\bm{D}}$:
        \begin{equation}
            \tilde{\bm{D}}_{2:N+1}^{-1} \tilde{\bm{D}}_1 = \begin{bmatrix}
                -\bm{1}_{N-1}\\
                0
            \end{bmatrix}\label{eq:inverse}
        \end{equation}
        where $\bm{1}_{N-1}$ denotes a column vector of ones of dimension $N-1$.
    \end{enumerate}
\end{prop}
\begin{proof}
    For part (i), we establish the nonsingularity of $\tilde{\bm{D}}_{2:N+1}$ by contradiction. Suppose there exists a nonzero vector $\bm{V} \in \mathbb{R}^{N}$ such that $\tilde{\bm{D}}_{2:N+1}\bm{V} = \bm{0}$. Consider the polynomial $v(\tau) = \sum_{i=2}^{N} V_{i-1}L_i(\tau) + V_N L_p(\tau)$. The condition $\tilde{\bm{D}}_{2:N+1}\bm{V} = \bm{0}$ implies that $\dot{v}(\tau_i) = 0$ for $i = 1,2,\ldots,N$. Since $\dot{v}(\tau)$ is a polynomial of degree at most $N-1$ with $N$ distinct zeros, it must be identically zero. This implies that $v(\tau)$ is a constant polynomial. However, by construction, $v(\tau_1) = 0$ at the boundary point $\tau_1 = -1$, which means $v(\tau) \equiv 0$. This contradicts our assumption that $\bm{V}$ is nonzero, thereby establishing that $\tilde{\bm{D}}_{2:N+1}$ is nonsingular.

    For part (ii), we note that the rows of $\tilde{\bm{D}}$ represent the derivatives of polynomials evaluated at the nodes. Since the constant polynomial $p(\tau) = 1$ has zero derivative everywhere, we have $\tilde{\bm{D}} \begin{bmatrix}
        \bm{1}_{N}\\
        0
    \end{bmatrix} = \bm{0}$. This can be rewritten as $\tilde{\bm{D}}_1 + \tilde{\bm{D}}_{2:N+1}\begin{bmatrix}\bm{1}_{N-1}\\0\end{bmatrix} = \bm{0}$, which implies $\tilde{\bm{D}}_{2:N+1}^{-1}\tilde{\bm{D}}_1 = \begin{bmatrix}-\bm{1}_{N-1}\\0\end{bmatrix}$, completing the proof.
\end{proof}

\subsection{Adjoint system in augmented differential form}
To derive the first-order necessary optimality conditions for the discretized problem formulated in \eqref{eq:discretized_cost}--\eqref{eq:discretized_boundary}, we define the Lagrangian function as follows
\begin{equation}
    \begin{aligned}
        \mathcal{L} = \Phi\left(\bm{X}_N\right) + \langle \bm{w}, \bm{G}(\bm{X}, \bm{U}) \rangle + \langle \bm{\mu}, \bm{x}_0 - \bm{X}_1\rangle  
        + \langle \bm{\Lambda}, \bm{F}(\bm{X}, \bm{U}) - \bm{D} \bm{X} - \bm{D}_p \bm{X}_p\rangle
    \end{aligned}
\end{equation}
where $\bm{\Lambda} \in \mathbb{R}^{N \times n}$ is the matrix of Lagrange multipliers associated with the discretized dynamics constraints in \eqref{eq:discretized_dynamics}, and $\bm{\mu} \in \mathbb{R}^{n}$ is the vector of Lagrange multipliers for the initial boundary condition constraint in \eqref{eq:discretized_boundary}. 

The first-order necessary conditions for optimality (the Karush-Kuhn-Tucker (KKT) conditions) of the discretized problem \eqref{eq:discretized_cost}--\eqref{eq:discretized_boundary} are derived by differentiating the Lagrangian with respect to each variable. This yields:
\begin{align}
    & \nabla_{\bm{X}}\left\langle \bm{w}_1, \bm{G}\left(\bm{X}_1, \bm{U}_1\right)\right\rangle + \nabla_{\bm{X}}\left\langle\bm{\Lambda}_1, \bm{F}\left(\bm{X}_1, \bm{U}_1\right)\right\rangle = \bm{D}_1^\mathrm{T}\bm{\Lambda} + \bm{\mu}\label{eq:adjoint_boundary}\\
    & \nabla_{\bm{X}}\left\langle\bm{w}_{2: N-1}, \bm{G}\left(\bm{X}_{2: N-1}, \bm{U}_{2: N-1}\right)\right\rangle  + \nabla_{\bm{X}}\left\langle\bm{\Lambda}_{2: N-1}, \bm{F}\left(\bm{X}_{2: N-1}, \bm{U}_{2: N-1}\right)\right\rangle = \bm{D}_{2: N-1}^\mathrm{T}\bm{\Lambda}\\
    & \nabla_{\bm{X}}\left\langle\bm{w}_N, \bm{G}\left(\bm{X}_N, \bm{U}_N\right)\right\rangle + \nabla_{\bm{X}}\left\langle\bm{\Lambda}_N, \bm{F}\left(\bm{X}_N, \bm{U}_N\right)\right\rangle = \bm{D}_N^\mathrm{T}\bm{\Lambda} - \nabla_{\bm{X}} \Phi\left(\bm{X}_N\right)\\
    & \nabla_{\bm{U}}\left\langle \bm{w}, \bm{G}\left(\bm{X}, \bm{U}\right)\right\rangle + \nabla_{\bm{U}}\left\langle\bm{\Lambda}, \bm{F}\left(\bm{X}, \bm{U}\right)\right\rangle = \bm{0} \label{eq:adjoint_control}\\
    & \bm{D}_p^\mathrm{T}\bm{\Lambda} = \bm{0}\label{eq:adjoint_order_condition}
\end{align}
The final equation \eqref{eq:adjoint_order_condition} emerges specifically from the introduction of the additional degree of freedom $\bm{X}_p$ in our augmented formulation.

Following the same arrangement and approach as established in \citet{gargUnifiedFrameworkNumerical2010a}, we transform the adjoint system using the relationship between the Lagrange multipliers $\bm{\Lambda}$ and the costate approximation $\bm{\lambda}$ at the collocation points:
\begin{align}
    \bm{\lambda}_i = \bm{\Lambda}_i / w_i, \quad i = 1, 2, \dots, N
\end{align}
where $w_i$ represents the quadrature weight associated with the $i$-th LGL node. With this transformation, the adjoint system given by Equations~\eqref{eq:adjoint_boundary}--\eqref{eq:adjoint_order_condition} can be expressed in the following form:
\begin{align}
    &\bm{D}^\dagger\bm{\lambda}= -\nabla_{\bm{X}}\bm{H}\left(\bm{X}, \bm{\lambda}, \bm{U}\right) + \frac{\bm{e}_1}{w_1}\left(\bm{\mu}-\bm{\lambda}_1\right)+\frac{\bm{e}_N}{w_N}\left(\bm{\lambda}_N-\nabla_{\bm{X}} \Phi\left(\bm{X}_N\right)\right)\\
    & \nabla_{\bm{U}}\bm{H}\left(\bm{X}, \bm{\lambda}, \bm{U}\right) = \bm{0}\\
    & \bm{D}_p^\mathrm{T}\bm{W}\bm{\lambda} = \bm{0}\label{eq:adjoint_order_condition_transformed}
\end{align}
where $\bm{e}_1$ and $\bm{e}_N$ are the first and last columns of the $N \times N$ identity matrix, respectively, $\bm{W}$ is the diagonal matrix with the LGL weights as the diagonal elements, and $\bm{H}$ is the Hamiltonian function $\mathcal{H}$ evaluated at the LGL nodes. The matrix $\bm{D}^\dagger$ is defined as
\begin{align}
    \bm{D}^\dagger = -\bm{W}^{-1} \bm{D}^\mathrm{T} \bm{W} -\frac{\bm{e}_1 \bm{e}_1^\mathrm{T}}{w_1} + \frac{\bm{e}_N \bm{e}_N^\mathrm{T}}{w_N}
\end{align}
It has been shown by \citet{fahrooDiscreteTimeOptimalityConditions2006} that $\bm{D} = \bm{D}^\dagger$.

The additional constraint in equation \eqref{eq:adjoint_order_condition_transformed} has a specific mathematical interpretation, which is formally established by the following proposition.

\begin{prop}\label{prop:order_condition}
    For any vector $\bm{V} \in \mathbb{R}^N$, the condition $\bm{D}_p^\mathrm{T} \bm{W}\bm{V} = 0$ holds if and only if the interpolation polynomial $v(\tau) = \sum_{i=1}^{N} \bm{V}_i L_i(\tau)$ is a polynomial of degree at most $N-2$. Equivalently, this condition ensures that the highest-order coefficient of the interpolation polynomial vanishes.
\end{prop}

\begin{proof}
    Suppose that the interpolation polynomial $v(\tau)$ is a polynomial of degree at most $N-2$. The product $v(\tau)\dot{L}_p(\tau)$ is a polynomial of degree at most $2N-3$, which can be integrated exactly by the LGL quadrature rule. Applying integration by parts, we obtain
    \begin{align}
        \bm{D}_p^\mathrm{T} \bm{W} \bm{V} &= \int_{-1}^1 v(\tau) \dot{L}_p(\tau)\,\mathrm{d}\tau\\
        &= \left. v(\tau) L_p(\tau)\right|_{-1}^1 - \int_{-1}^1 \dot{v}(\tau) L_p(\tau)\,\mathrm{d}\tau\\
        &= \left.v(\tau) L_p(\tau)\right|_{-1}^1 - \sum_{i=1}^N w_i \dot{v}(\tau_i) L_p(\tau_i)\\
        &= 0
    \end{align} 
    The final equality follows since $L_p(\tau_i) = 0$ for all LGL nodes by construction, which causes both terms to vanish.

    Conversely, assume $\bm{D}_p^\mathrm{T} \bm{W} \bm{V} = 0$. By the properties of Lagrange interpolation, there exists a unique vector $\bm{V}' \in \mathbb{R}^N$ such that $\bm{V}'_i = \bm{V}_i$ for $i=1,\ldots,N-1$, and the polynomial $v'(\tau) = \sum_{i=1}^{N} \bm{V}'_i L_i(\tau)$ is of degree at most $N-2$. From the first part of this proof, we know that $\bm{D}_p^\mathrm{T} \bm{W} \bm{V}' = 0$. 
    Therefore, $\bm{D}_p^\mathrm{T} \bm{W}(\bm{V} - \bm{V}') = 0$. Since neither the last element of the weight vector $\bm{w}$ nor the last element of $\bm{D}_p$ is zero, the last component of the vector $\bm{V} - \bm{V}'$ must be zero. Consequently, $\bm{V} = \bm{V}'$, which proves that $v(\tau)$ is indeed a polynomial of degree at most $N-2$.
\end{proof}

According to Proposition~\ref{prop:order_condition}, we can establish that the adjoint system derived from the discretized problem \eqref{eq:discretized_cost}--\eqref{eq:discretized_boundary} corresponds precisely to a Lobatto IIIB discontinuous collocation method for the costate vector $\bm{\lambda}$ in the continuous costate equations \eqref{eq:costate_boundary}--\eqref{eq:costate_control}~\cite{ernsthairerGeometricNumericalIntegration2006}. This correspondence demonstrates that our augmented formulation produces a mathematically consistent pseudospectral approximation to the dynamics of costate variables in the necessary conditions of optimality for the continuous optimal control problem. By introducing the additional degree of freedom, we effectively eliminate the null space present in the adjoint system of the standard LGL collocation method. This directly addresses the fundamental limitations previously identified by \citet{gargUnifiedFrameworkNumerical2010a} and establishes a more complete theoretical foundation for LGL-based pseudospectral methods in optimal control.

\section{Integral formulation using LGL collocation}\label{sec:integral_formulation}
\subsection{Equivalence between the augmented and integral formulations}
This section demonstrates the equivalence between the previously developed augmented LGL collocation method and the integral formulation of the LGL collocation method~\cite{gargUnifiedFrameworkNumerical2010a}. The key distinction is that in the integral formulation, the additional degree of freedom becomes implicit rather than explicit.

Leveraging Proposition~\ref{prop:inverse}, we can reformulate the discretized dynamics equation \eqref{eq:discretized_dynamics} into an equivalent integral form. By defining $\tilde{\bm{A}} = \tilde{\bm{D}}_{2:N+1}^{-1}$ and applying the inverse operation to equation \eqref{eq:discretized_dynamics}, we obtain:
\begin{align}
    \tilde{\bm{A}} \bm{F}(\bm{X}, \bm{U}) = \begin{bmatrix}
        -\bm{1}_{N-1}\\
        0
    \end{bmatrix} \bm{X}_1 + \tilde{\bm{X}}_{2:N+1}\label{eq:integral_dynamics}
\end{align}
Equation~\eqref{eq:integral_dynamics} can be decomposed into two components:
\begin{align}
    \bm{X}_{2:N} &= \bm{1}_{N-1} \bm{X}_1 + \bm{A} \bm{F}(\bm{X}, \bm{U})\label{eq:integral_dynamics_split1}\\
    \bm{X}_p &= \bm{A}_p \bm{F}(\bm{X}, \bm{U})\label{eq:integral_dynamics_split2}
\end{align}
where $\bm{A}$ comprises the first $N-1$ rows of $\tilde{\bm{A}}$, and $\bm{A}_p$ represents its final row. This decomposition reveals how the state variables and the augmentation coefficient are individually determined through the integration process. For consistency with the indexing throughout the paper, we denote $\bm{A}_i$ as the $(i-1)$th row of the matrix $\bm{A}$, where $i = 2, \ldots, N$.

The relationship between the augmented LGL collocation method and the integral formulation of the LGL collocation method is established through the following proposition.

\begin{prop}\label{prop:equivalence}
    The matrices $\bm{A}$ and $\bm{A}_p$ in equations~\eqref{eq:integral_dynamics_split1} and \eqref{eq:integral_dynamics_split2}, respectively, have the following properties:
    \begin{enumerate}[label=(\roman*)]
    \item The integral formulation of the LGL collocation method is mathematically equivalent to equation~\eqref{eq:integral_dynamics_split1}. Specifically, the following equation holds
    \begin{align}
        \bm{A}_{i, j} = \int_{-1}^{\tau_i} L_j(\tau) \,\mathrm{d}\tau\label{eq:integral_matrix}
    \end{align}
    for $i = 2, \dots, N$ and $j = 1, \dots, N$, where $\bm{A}_{i,j}$ is the $(i,j)$-th element of the matrix $\bm{A}$. 
    
    \item Furthermore,
    \begin{align}
        \left(\bm{A}_p\right)_i = \frac 1N \prod_{j=1, j\neq i}^{N}\frac{1}{\tau_i - \tau_j}\label{eq:integral_vector}
    \end{align}
    for $i = 1, \dots, N$, where $\left(\bm{A}_p\right)_i$ is the $i$-th element of the vector $\bm{A}_p$.
    \end{enumerate}
\end{prop}
\begin{proof}
    The value of $\left(\bm{A}_p\right)_i$ corresponds to the coefficient of the $x^N$ term in the polynomial resulting from the integral $\int L_i(\tau)\,\mathrm d\tau$. Using the values of $\bm{A}$ in equation~\eqref{eq:integral_matrix} and $\bm{A}_p$ in equation~\eqref{eq:integral_vector}, we can directly verify that
    \begin{align}
        \tilde{\bm{D}}\begin{bmatrix}
            \bm{0}_N\\
            \bm{A}\\
            \bm{A}_p
        \end{bmatrix} = \bm{I}_N
    \end{align}
    which is equivalent to
    \begin{align}
        \tilde{\bm{A}} = \begin{bmatrix}
            \bm{A}\\
            \bm{A}_p
        \end{bmatrix} = \tilde{\bm{D}}_{2:N+1}^{-1}
    \end{align}
\end{proof}

A notable advantage of the integral formulation of the LGL collocation method is that the auxiliary variables $\bm{X}_p$ are fully determined by equation~\eqref{eq:integral_dynamics_split2} and do not influence any other equation in the system. Consequently, $\bm{X}_p$ need not be explicitly represented as variables in the corresponding NLP problem. This implicit representation effectively reduces the dimensionality of the optimization problem without sacrificing computational accuracy, providing computational advantages that make the LGL approach particularly efficient for practical implementations, as will be further discussed in Section~\ref{sec:comp_n}.

\subsection{Adjoint system in integral form}
In this section, we derive the adjoint system for the integral formulation of the LGL collocation method. The following derivation demonstrates that this adjoint system still corresponds to a Lobatto IIIB discontinuous collocation method for the costate vector $\bm{\lambda}$ with respect to the continuous costate equations \eqref{eq:costate_boundary}--\eqref{eq:costate_control}~\cite{ernsthairerGeometricNumericalIntegration2006}. 

The Lagrangian function for the integral formulation of the LGL collocation method is given by
\begin{equation}
        \mathcal{L} = \Phi\left(\bm{X}_N\right) + \langle \bm{w}, \bm{G}(\bm{X}, \bm{U})\rangle + \langle \bm{\mu}, \bm{x}_0 - \bm{X}_1\rangle     + \langle \bm{R}, \bm{A}\bm{F}(\bm{X}, \bm{U}) -\bm{X}_{2: N} + \bm{1}_{N-1} \bm{X}_1\rangle
\end{equation}
where $\bm{R} \in \mathbb{R}^{(N-1) \times n}$ represents the matrix of Lagrange multipliers associated with the discretized dynamics constraints in \eqref{eq:integral_dynamics_split1}, and $\bm{\mu} \in \mathbb{R}^{n}$ is the Lagrange multiplier for the initial boundary condition constraint. For consistency with the notation throughout this paper, we denote $\bm{R}_i$ as the $(i-1)$-th row of the matrix $\bm{R}$, where $i = 2, \ldots, N$.

Applying the KKT conditions to the Lagrangian function, we obtain the following system of equations for the adjoint system
\begin{align}
    & \nabla_{\bm{X}}\left\langle \bm{w}_1, \bm{G}\left(\bm{X}_1, \bm{U}_1\right)\right\rangle  + \nabla_{\bm{X}}\left\langle\bm{A}^{\mathrm{T}}_1\bm{R}, \bm{F}\left(\bm{X}_1, \bm{U}_1\right)\right\rangle = - \bm{1}_{N-1}^\mathrm{T}\bm{R} + \bm{\mu}\label{eq:adjoint_boundary_integral}\\
    & \nabla_{\bm{X}}\left\langle \bm{w}_{2:N-1}, \bm{G}\left(\bm{X}_{2:N-1}, \bm{U}_{2:N-1}\right)\right\rangle + \nabla_{\bm{X}}\left\langle \bm{A}^{\mathrm{T}}_{2:N-1}\bm{R}, \bm{F}\left(\bm{X}_{2:N-1}, \bm{U}_{2:N-1}\right)\right\rangle = \bm{R}_{2:N-1}\\
    & \nabla_{\bm{X}}\left\langle \bm{w}_N, \bm{G}\left(\bm{X}_N, \bm{U}_N\right)\right\rangle+ \nabla_{\bm{X}}\left\langle \bm{A}^{\mathrm{T}}_N\bm{R}, \bm{F}\left(\bm{X}_N, \bm{U}_N\right)\right\rangle = \bm{R}_N  - \nabla_{\bm{X}} \Phi\left(\bm{X}_N\right) \\
    & \nabla_{\bm{U}}\left\langle \bm{w}, \bm{G}\left(\bm{X}, \bm{U}\right)\right\rangle + \nabla_{\bm{U}}\left\langle\bm{A}^{\mathrm{T}}\bm{R}, \bm{F}\left(\bm{X}, \bm{U}\right)\right\rangle = \bm{0} \label{eq:adjoint_control_integral}
\end{align}

Similarly to the differential form, set
\begin{align}
    \bm{r}_{i} &= \bm{R}_i / w_i, \quad i = 2, \dots, N
\end{align}
where $w_i$ is the quadrature weight associated with the $i$-th LGL node. The adjoint system can then be expressed in the following form:
\begin{align}
    & \nabla_{\bm{X}}\left\langle \bm{1}_{N}, \bm{G}\left(\bm{X}, \bm{U}\right)\right\rangle + \nabla_{\bm{X}}\left\langle\bm{A}^\dagger \bm{r}, \bm{F}\left(\bm{X}, \bm{U}\right)\right\rangle = \begin{bmatrix}
        -\bm{w}_{2:N}^\mathrm{T} /w_1 \\
        \bm{I}_{N-1} 
    \end{bmatrix} \bm{r} + \frac{\bm{e}_1}{w_1}\bm{\mu} - \frac{\bm{e}_N}{w_N}\nabla_{\bm{X}} \Phi\left(\bm{X}_N\right)\label{eq:adjoint_boundary_integral_sub}\\
    & \nabla_{\bm{U}}\left\langle\bm{1}_{N}, \bm{G}\left(\bm{X}, \bm{U}\right)\right\rangle + \nabla_{\bm{U}}\left\langle\bm{A}^\dagger\bm{r}, \bm{F}\left(\bm{X}, \bm{U}\right)\right\rangle = \bm{0} \label{eq:adjoint_control_integral_sub}
\end{align}
where $\bm{w}_{2:N}$ is the vector of weights associated with the LGL nodes from $2$ to $N$, and $\bm{1}_{N}$ is the column vector of ones of dimension $N$. The matrix $\bm{A}^\dagger$ is defined as
\begin{align}
    \bm{A}^\dagger = \bm{W}^{-1} \bm{A}^\mathrm{T} \bm{W}_{2:N}
\end{align}
where $\bm{W}_{2:N}$ is the diagonal matrix with the LGL weights from indices $2$ to $N$ as its diagonal elements.

\begin{prop}\label{prop:adjoint_integral_matrix}
    The adjoint integral matrix $\bm{A}^\dagger$ can be viewed as an integration matrix with respect to interpolation polynomials $M_i(\tau)$ defined for the second to the $(N-1)$-th LGL nodes, with the integration constant controlled by the value of the variable at the last LGL node. Specifically, the elements of the adjoint integral matrix $\bm{A}^\dagger$ are given by
    \begin{align}
        \bm{A}^\dagger_{i, j} &= \int_{\tau_i}^1 M_j(\tau) \,\mathrm{d}\tau - w_N M_j(1),\quad i = 1, \ldots, N, \quad j = 2, \ldots, N-1\label{eq:adjoint_integral_matrix_body}\\
        \bm{A}^\dagger_{i, N} &= w_N, \quad i = 1, \ldots, N\label{eq:adjoint_integral_matrix_tail}
    \end{align}
    where $M_i(\tau), i = 2, \dots, N-1$ are the $(N-3)$-th degree polynomials satisfying $M_i(\tau_j) = \delta_{ij}$, and $\delta_{ij}$ is the Kronecker delta function.
\end{prop}

\begin{proof}
    For $j = 2, \dots, N-1$, we have
    \begin{align}
        \int_{-1}^1 M_j(\tau) \int_{-1}^{\tau}L_i(t)\,\mathrm{d}t \,\mathrm{d}\tau = \int_{-1}^1 L_i(t) \int_{t}^1 M_j(\tau) \,\mathrm{d}\tau \,\mathrm{d}t
    \end{align}
    The integrand of the outer integral is a polynomial of degree at most $N-3$, which can be integrated exactly by the LGL quadrature rule. Therefore
    \begin{align}
        w_j A_{j, i} + w_N M_j(1) w_i = w_i \int_{\tau_i}^1 M_j(\tau)\,\mathrm{d}\tau
    \end{align}
    which is equivalent to equation~\eqref{eq:adjoint_integral_matrix_body}.
    
    The last equation \eqref{eq:adjoint_integral_matrix_tail} follows directly since $A_{N, i} = w_i$ for $i = 1, \ldots, N$.
\end{proof}

Given that the costate variable $\bm{\lambda}$ is discretized and parameterized by the product $\bm{A}^\dagger \bm{r}$, we can directly verify that the adjoint system given by equations~\eqref{eq:adjoint_boundary_integral_sub}--\eqref{eq:adjoint_control_integral_sub} corresponds to a Lobatto IIIB discontinuous collocation method for the costate vector $\bm{\lambda}$~\cite{ernsthairerGeometricNumericalIntegration2006}, as formalized in the following proposition.

\begin{prop}\label{prop:adjoint_integral_matrix_order}
    The adjoint system given by equations~\eqref{eq:adjoint_boundary_integral_sub}--\eqref{eq:adjoint_control_integral_sub} represents a Lobatto IIIB discontinuous collocation method for the costate vector $\bm{\lambda}$ with respect to the continuous costate equations \eqref{eq:costate_boundary}--\eqref{eq:costate_control}.
\end{prop}
\begin{proof}
    Firstly, for any vector $\bm{r} \in \mathbb{R}^{N-1}$, consider the interpolation polynomial defined by
        \begin{align}
            \lambda(\tau) = \sum_{i=1}^{N} (\bm{A}^\dagger_i\bm{r}) L_i(\tau)
        \end{align}
        where $\bm{A}^\dagger_i$ is the $i$-th row of the matrix $\bm{A}^\dagger$. This polynomial is of degree at most $N-2$,
        a property that follows directly from Proposition~\ref{prop:adjoint_integral_matrix}, as the polynomials $M_i(\tau)$ are of degree $N-3$.

        Furthermore, the derivative of the polynomial $\lambda(\tau)$ at the LGL nodes can be expressed as
        \begin{align}
            \dot{\lambda}(\tau_1) &= -\sum_{i=2}^{N-1} M_i(-1) \bm{r}_i\label{eq:integral_parametrization_first}\\
            \dot{\lambda}(\tau_i) &= -\bm{r}_i, \quad i = 2, \ldots, N-1\label{eq:integral_parametrization_middle}\\
            \dot{\lambda}(\tau_N) &= -\sum_{i=2}^{N-1} M_i(1) \bm{r}_i\label{eq:integral_parametrization_last}
        \end{align}
        By comparing equation~\eqref{eq:integral_parametrization_middle} with equations~\eqref{eq:costate_dynamics} and \eqref{eq:adjoint_boundary_integral_sub}, we can verify that the collocation conditions at the interior LGL nodes (from $i=2$ to $i=N-1$) are satisfied by the adjoint system as required.

        For the first and last LGL nodes, the discontinuous collocation requires 
        \begin{align}
            \nabla_{\bm{X}}\mathcal{H}_1 &= -\dot{\lambda}(-1) + \frac{\mu - \lambda(-1)}{w_1} \label{eq:integral_collocation_first}\\
            \nabla_{\bm{X}}\mathcal{H}_N &= -\dot{\lambda}(1) + \frac{\lambda(1) - \nabla_{\bm{X}} \Phi}{w_N}\label{eq:integral_collocation_last}
        \end{align}
        The right-hand side of equation~\eqref{eq:integral_collocation_first} equals
        \begin{align}
            -\dot{\lambda}(-1) + \frac{\mu - \lambda(-1)}{w_1} 
            =& -\sum_{i=2}^{N-1} M_i(-1) \bm{r}_i + \frac{\mu - \sum_{i=2}^{N-1}\left(w_1M_i(-1)+w_i\right)\bm{r}_i-w_N\bm{r}_N}{w_1}\\
            = &-\sum_{i=2}^{N} \frac{w_i}{w_1} \bm{r}_i + \frac{\mu}{w_1} 
        \end{align}
        while the right-hand side of equation~\eqref{eq:integral_collocation_last} equals
        \begin{align}
            -\dot{\lambda}(1) + \frac{\lambda(1) - \nabla_{\bm{X}} \Phi}{w_N} 
            &= -\sum_{i=2}^{N-1} M_i(1) \bm{r}_i + \frac{-\sum_{i=2}^{N-1}w_NM_i(1)\bm{r}_i +w_N\bm{r}_N - \nabla_{\bm{X}} \Phi}{w_N}\\
            &= \bm{r}_N - \frac{\nabla_{\bm{X}} \Phi}{w_N}
        \end{align}
        These results exactly match the corresponding terms in equation~\eqref{eq:adjoint_boundary_integral_sub}, completing the proof.
\end{proof}

\section{Comparative Analysis of LG, LGR, and LGL Collocation Methods}\label{sec:comparison}

This section presents a comparison between the integral formulation of the LGL collocation method and the established LG and LGR collocation methods in a multiple-subinterval setting, highlighting the computational efficiency advantages of the LGL approach in practical applications.

For complex optimal control problems, the temporal domain is typically partitioned into $M$ subintervals, each implementing a pseudospectral collocation method to enhance the numerical performance and solution accuracy. This partitioning strategy is essential for several reasons:
\begin{enumerate}[label=(\roman*)]
    \item \textbf{Numerical conditioning:} Excessively high-degree polynomial approximations are prone to ill-conditioning phenomena. Domain decomposition effectively mitigates these issues by limiting the polynomial degree within each subinterval while maintaining overall solution accuracy for long-horizon problems. 
    \item \textbf{Representation of non-smooth solutions:} Many practical optimal control problems involve discontinuous control inputs. In such cases, global high-degree polynomial approximations suffer from the Gibbs phenomenon, resulting in oscillations near discontinuities. Localized lower-degree approximations across multiple subintervals provide superior representation of these non-smooth features~\cite{arribasOptimizationPathConstrainedSystems2015}.
    \item \textbf{Adaptive refinement:} Domain decomposition facilitates adaptive mesh refinement strategies, allowing computational resources to be concentrated in regions with complex dynamics or rapid state transitions~\cite{doi:10.2514/1.52136,darbyHpadaptivePseudospectralMethod2011}.
\end{enumerate}

Although this paper derives the covector mapping for a single subinterval, the results extend naturally to multiple subintervals. The differential and integral matrices for the multiple-subinterval case are constructed by placing the single-subinterval matrices along the diagonal of the overall system matrix. The theoretical properties and convergence results established for the single-subinterval case remain valid for the multiple-subinterval formulation.

\subsection{Problem Dimension}\label{sec:comp_n}

Assuming $N$ collocation points are used uniformly across $M$ subintervals, and that each collocation and non-collocation variable corresponds to an independent variable in the NLP problem, Table~\ref{tab:comp_n} summarizes the total number of variables for one state variable across the LG, LGR, and LGL methods. The number of collocation variables corresponds to both the number of dynamics function evaluations and the number of equality constraints needed to enforce the dynamics in the discretized problem. The number of non-collocation variables directly influences the total variable count in the NLP problem, which significantly impacts computational efficiency during optimization.

As shown in the table, the LGL method demonstrates advantages over both the LG and LGR methods in terms of collocation variables and total variable count. Regarding collocation variables, the LGL method's inclusion of two-sided boundary points allows these points to be shared between adjacent subintervals, thereby reducing the overall number of collocation variables. Regarding non-collocation variables, as established in the previous section, the auxiliary variable $\bm{X}_p$ can be made implicit and eliminated from the NLP problem, resulting in zero non-collocation variables for the LGL method. These advantages in both collocation and non-collocation variables combine to yield a substantial reduction in the total number of variables in the NLP problem compared to the other methods.

It is worth noting that the integration order of the LGL method is $2N-3$, while the LG and LGR methods achieve $2N-1$ and $2N-2$, respectively. To account for this difference, one could increase the number of collocation points in the LGL method to $N+1$, making it comparable to the LG method and superior to the LGR method in terms of integration order. With this adjustment, the collocation variable counts across all three methods become similar. Nevertheless, the LGL method still maintains its advantage over the LG method regarding non-collocation variables. Compared to the LGR method, the LGL approach maintains a similar number of non-collocation variables while delivering a higher integration order with minimal additional computational cost. In contrast to previous conclusions in the literature, we find that the LGL method offers a remarkable advantage in terms of the total number of variables in the NLP proble, which is a critical factor for computational efficiency during the solution process.

\begin{table}
    \centering
    \caption{Comparison of variable counts in the NLP problem across collocation methods. \#C: number of collocation variables, \#NC: number of non-collocation variables, \#ALL: total number of variables.}
    \begin{tabular}{cccc}
        \toprule
        Method & LG & LGR & LGL \\
        \midrule
        \#C & $MN$ & $MN$ & $M(N-1)+1$ \\
        \#NC & $M+1$ & $1$ & $0$ \\
        \#ALL & $M(N+1)+1$ & $MN+1$ & $M(N-1)+1$\\
        \bottomrule
    \end{tabular}\label{tab:comp_n}
\end{table}

\subsection{Symplecticity}\label{sec:comp_s}
It is well known that the first-order optimality conditions of optimal control problems for the state and costate variables satisfy the Hamiltonian canonical equations
\begin{align}
    \dot{\bm{x}} = \nabla_{\bm{\lambda}} \mathcal{H}, \quad \dot{\bm{\lambda}} = -\nabla_{\bm{x}} \mathcal{H}
\end{align}
As a result, for problems with a long time horizon, it is beneficial to employ symplectic integrators to preserve the Hamiltonian structure of the system. The analysis in sections~\ref{sec:augmented_formulation} and \ref{sec:integral_formulation} demonstrates that when the state variable is discretized using the augmented or integral LGL collocation method, the adjoint system for the costate variable is discretized using a Lobatto IIIB discontinuous collocation method. Since the Lobatto IIIA-IIIB pair is a symplectic integrator~\cite{ernsthairerGeometricNumericalIntegration2006}, the discretized NLP problem automatically forms a symplectic integrator for the Hamiltonian system. The situation is the same for the LG collocation method, where the adjoint system is also Gauss collocation. However, to the best of our knowledge, there is no evidence that the LGR discretization method for the state variable and its adjoint system for the costate variable is symplectic. Numerical experiments also indicate that the Hamiltonian value is not preserved for the LGR method, as shown in the example in section~\ref{sec:example2}. Within the context of preserving the symplectic structure, the LGL method is superior to the LGR method, which plays a crucial role in long-horizon trajectory optimization problems. A concrete example is presented in section~\ref{sec:example2}.

\section{Numerical Examples}\label{sec:examples}
\subsection{Simple Example Testing the Convergence of the Augmented Collocation Method}\label{sec:example1}
To demonstrate the effectiveness of the proposed augmented LGL collocation method, we reproduce the numerical experiment originally presented in \citet{gargUnifiedFrameworkNumerical2010a}. The optimal control problem is formulated as
\begin{align}
    \min_{y, u} & \quad -y(2)\\
    \text{s.t.} & \quad \dot{y} = \frac{5}{2}\left(-y+yu-u^2\right)\\
    & \quad y(0) = 1
\end{align}
The discretized nonlinear programming problem is solved using the interior-point optimizer IPOPT~\cite{wachterImplementationInteriorpointFilter2006} with convergence tolerance set to $10^{-13}$.

Figure~\ref{fig:collocation_errors} shows the $L_\infty$ norm of errors in the state, control, and costate variables at the collocation points, where Gauss and Radau refer to the LG and LGR methods in differential form, respectively, and Lobatto refers to the proposed augmented LGL method. In contrast to the results reported in \citet{gargUnifiedFrameworkNumerical2010a}, the error magnitudes for the augmented LGL method are comparable to those of the LG and LGR methods for both state and control variables. While the costate variable initially shows higher error values in the LGL implementation, these errors converge to levels similar to those of the other methods, being limited primarily by the numerical precision of the optimization solver. These results demonstrate that the proposed augmented LGL collocation method successfully addresses the fundamental limitations of the standard LGL approach, ensuring proper convergence of the adjoint system to the continuous costate equations.

\begin{figure}
    \centering
    \includegraphics[width=0.8\columnwidth]{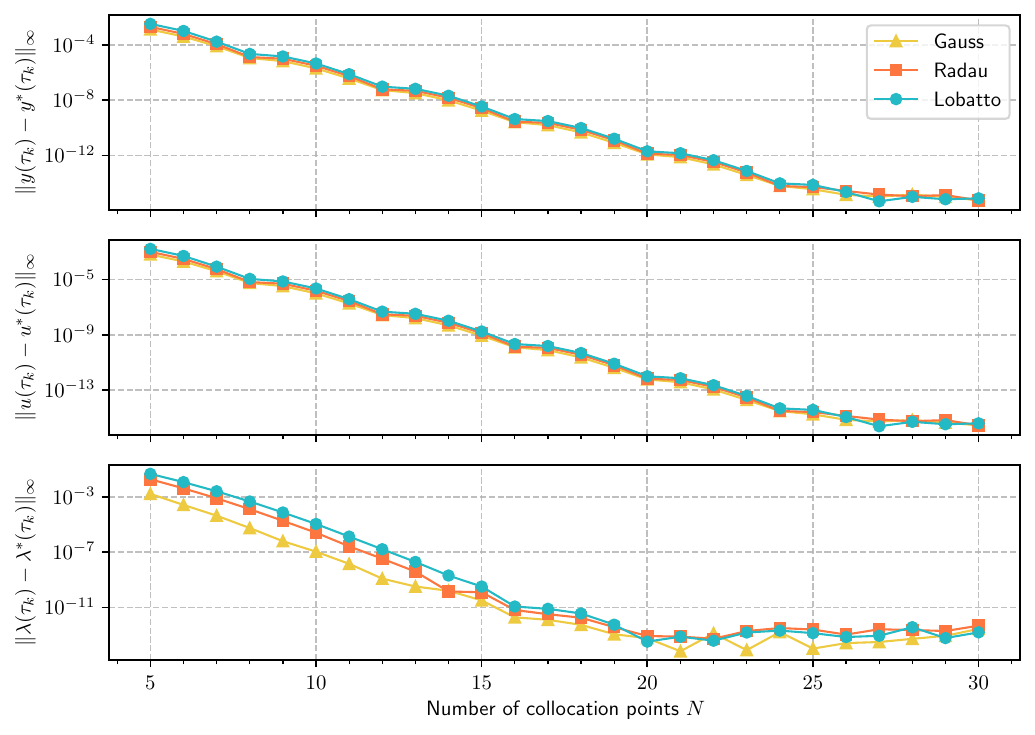}
    \caption{$L_\infty$ norm of errors in state, control, and costate variables at the collocation points.}
    \label{fig:collocation_errors}
\end{figure}

\subsection{Long-Horizon Oscillatory Problem Testing the Symplecticity Property}\label{sec:example2}
To test the symplecticity property of the LGL method, we consider a long-horizon oscillatory problem: the multiple-revolution low-thrust trajectory optimization problem for spacecraft, which is of vital importance in aerospace engineering. The problem involves four state variables $[p, f, g, l]$ and two control variables $[a_r, a_t]$.
The optimization problem with its dynamics is given by
\begin{align}
    \min_{a_r, a_t} & \quad \int_0^T \left(a_r^2 + a_t^2\right) \,\mathrm{d}t\\
    \text{s.t.} & \quad \dot{p} = \frac{2 p a_t}{w} \sqrt{\frac {p}{\mu}} \quad\label{eq:dotp}\\
    & \quad\dot{f} = \sqrt{\frac{p}{\mu}} \left\{  a_r \sin l +\frac{\left[(w+1)\cos l+f\right]a_t}{w}\right\}\\
    & \quad\dot{g} = \sqrt{\frac{p}{\mu}} \left\{  -a_r \cos l +\frac{\left[(w+1)\sin l+g\right]a_t}{w}\right\}\\
    & \quad\dot{l} = \sqrt{\mu p} \left(\frac wp\right)^2\label{eq:dotL}
\end{align}
where $w = 1 + f\cos l + g\sin l$, and $\mu$ is the gravitational parameter. 

The initial and final conditions are given by
\begin{align}
    &p(0) = 9128\, \text{km}, \quad f(0) = g(0) = l(0) = 0\\
    &p(T) = 42164\, \text{km}, \quad f(T) = g(T) = 0, \quad l(T) = 250\pi
\end{align}
where the final time $T$ is free. The final time is treated as an optimization variable in the discretized problem to linearly scale the time interval. This free final time $T$ does not affect the parameterization of the state and control variables, and the analysis presented in this paper remains valid.

The problem is solved using our Python-based pseudospectral optimal control package \emph{pockit}\footnote{The package is free and open-source, available at \textsf{https://github.com/zouyilin2000/pockit}}. We solve the NLP problem using the IPOPT~\cite{wachterImplementationInteriorpointFilter2006} solver. The problem is discretized using both LGR and LGL methods in the integral formulation, with $M = 888$ subintervals and varying numbers of collocation points per subinterval. We set the convergence tolerance to $10^{-12}$, while keeping all other parameters at their default values. For computational purposes, we set the length unit as Earth's radius ($6378.1363 \text{ km}$) and the time unit as $1$ day. The gravitational parameter $\mu$ is set to $398600.4418 \text{ km}^3/\text{s}^2$.

Figure~\ref{fig:covector_hamiltonian} shows the Hamiltonian values for both LGR and LGL methods over time, with $N = 3$ collocation points per subinterval. The LGR method exhibits a lower mean Hamiltonian value initially and fails to oscillate around a constant value. In contrast, the LGL method's Hamiltonian oscillates around $0$, which aligns with our theoretical analysis that the LGL discretization method automatically forms a symplectic integrator for the Hamiltonian system. Additionally, from the calculus of variations and the PMP, we know that the Hamiltonian for the optimal trajectory should be identically $0$ throughout the entire time interval. Although the LGR method's Hamiltonian stops oscillating near the final time, it shows a slight deviation from $0$. Conversely, the LGL method demonstrates the desired behavior as its Hamiltonian value approaches $0$ with diminishing oscillations over time.

Figure~\ref{fig:covector_error} shows the absolute error of the objective function for both LGR and LGL methods across various numbers of collocation points $N$. The reference solution was obtained using the LGR method at significantly higher resolution ($M = 5555$ and $N = 12$) than the test cases. As illustrated in the figure, the LGL method achieves errors that are orders of magnitude lower than the LGR method, demonstrating the superior effectiveness of the LGL method for long-horizon problems.

\begin{figure}
    \centering
    \begin{subfigure}{0.48\textwidth}
        \centering
        \includegraphics[width=\textwidth]{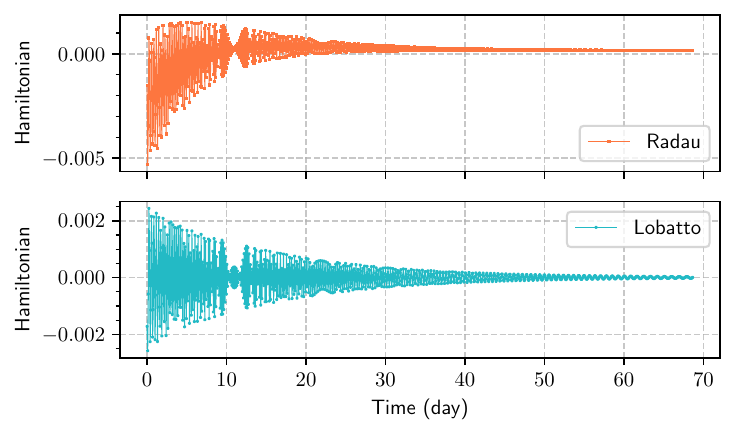}
        \caption{The Hamiltonian value over time, $N = 3$.}
        \label{fig:covector_hamiltonian}
    \end{subfigure}
    \hfill
    \begin{subfigure}{0.48\textwidth}
        \centering
        \includegraphics[width=\textwidth]{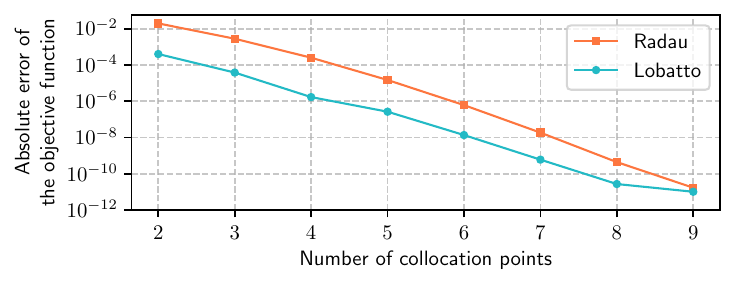}
        \caption{The absolute error of the objective function for $N$ ranging from $2$ to $9$.}
        \label{fig:covector_error}
    \end{subfigure}
    \caption{Numerical results for the long-horizon oscillatory problem.}
    \label{fig:covector_results}
\end{figure}

\section{Conclusion}\label{sec:conclusion}
Pseudospectral collocation methods are widely used in optimal control problems due to their accuracy and computational efficiency. Historically, the LGL collocation method has been considered less effective than the LG and LGR methods, primarily because of issues with establishing convergence of the adjoint system. This paper presents a thorough re-examination of the LGL collocation method.

We proposed an augmented formulation that addresses the limitations of the standard LGL approach. This augmented LGL collocation method is mathematically equivalent to the integral formulation of the LGL method, differing only in the implicit representation of the auxiliary variable in the integral formulation. Through analytical derivation, we established that the adjoint system for both the augmented differential and integral formulations corresponds to a Lobatto IIIB discontinuous collocation method for the costate vector. Our comparative analysis of the LG, LGR, and LGL methods demonstrated that the LGL approach requires fewer variables in the NLP problem, making it computationally more efficient for practical applications. We also examined the symplectic property of the LGL method, revealing its advantages for long-horizon problems where preserving the Hamiltonian structure is important. We validated our theoretical findings through two numerical examples. First, a simple test case demonstrating the convergence properties of the augmented LGL method, showing that the augmented formulation resolves the convergence issues of the standard LGL method reported in the literature. The state, control, and costate variables converge to the exact solution exponentially, similar to the LG and LGR methods. Second, a long-horizon oscillatory problem highlighting the symplectic advantages of the LGL approach. The results indicate that the LGL method preserves the Hamiltonian value as expected and achieves significantly lower errors in the objective function compared to the LGR method, demonstrating its superior performance for long-horizon trajectory optimization problems.

The results presented in this paper indicate that the LGL method is a robust and efficient choice for solving optimal control problems, particularly for long-horizon scenarios. This work provides a comprehensive reassessment of the LGL method's capabilities and establishes its competitive position among pseudospectral methods for practical optimal control applications.

\section*{Acknowledgments}
This work was supported by the National Natural Science Foundation of China (Grant No.12472355).

\bibliography{lobatto}

\begin{thebibliography}{22}
\providecommand{\natexlab}[1]{#1}
\providecommand{\url}[1]{\texttt{#1}}
\expandafter\ifx\csname urlstyle\endcsname\relax
  \providecommand{\doi}[1]{doi: #1}\else
  \providecommand{\doi}{doi: \begingroup \urlstyle{rm}\Url}\fi

\bibitem[Kelly(2017)]{kellyIntroductionTrajectoryOptimization2017}
Matthew Kelly.
\newblock An {{Introduction}} to {{Trajectory Optimization}}: {{How}} to {{Do Your Own Direct Collocation}}.
\newblock \emph{SIAM Rev.}, 59\penalty0 (4):\penalty0 849--904, January 2017.
\newblock ISSN 0036-1445.
\newblock \doi{10.1137/16M1062569}.

\bibitem[Wang et~al.(2019)Wang, Cui, and Wei]{wangOptimalRocketLanding2019}
Jinbo Wang, Naigang Cui, and Changzhu Wei.
\newblock Optimal {{Rocket Landing Guidance Using Convex Optimization}} and {{Model Predictive Control}}.
\newblock \emph{J. Guid. Control Dyn.}, 42\penalty0 (5):\penalty0 1078--1092, May 2019.
\newblock ISSN 0731-5090, 1533-3884.
\newblock \doi{10.2514/1.G003518}.

\bibitem[W{\"a}chter and Biegler(2006)]{wachterImplementationInteriorpointFilter2006}
Andreas W{\"a}chter and Lorenz~T. Biegler.
\newblock On the implementation of an interior-point filter line-search algorithm for large-scale nonlinear programming.
\newblock \emph{Math. Program.}, 106\penalty0 (1):\penalty0 25--57, March 2006.
\newblock ISSN 1436-4646.
\newblock \doi{10.1007/s10107-004-0559-y}.

\bibitem[Gill et~al.(2002)Gill, Murray, and Saunders]{gillSNOPTSQPAlgorithm2002}
Philip~E. Gill, Walter Murray, and Michael~A. Saunders.
\newblock {{SNOPT}}: {{An SQP Algorithm}} for {{Large-Scale Constrained Optimization}}.
\newblock \emph{SIAM J. Optim.}, 12\penalty0 (4):\penalty0 979--1006, January 2002.
\newblock ISSN 1052-6234.
\newblock \doi{10.1137/S1052623499350013}.

\bibitem[Patterson and Rao(2014)]{pattersonGPOPSIIMATLABSoftware2014a}
Michael~A. Patterson and Anil~V. Rao.
\newblock {{GPOPS-II}}: {{A MATLAB Software}} for {{Solving Multiple-Phase Optimal Control Problems Using}} hp-{{Adaptive Gaussian Quadrature Collocation Methods}} and {{Sparse Nonlinear Programming}}.
\newblock \emph{ACM Trans. Math. Softw.}, 41\penalty0 (1):\penalty0 1:1--1:37, October 2014.
\newblock ISSN 0098-3500.
\newblock \doi{10.1145/2558904}.

\bibitem[Becerra(2010)]{5612676}
Victor~M. Becerra.
\newblock Solving complex optimal control problems at no cost with {{PSOPT}}.
\newblock In \emph{2010 {{IEEE}} International Symposium on Computer-Aided Control System Design}, pages 1391--1396, 2010.
\newblock \doi{10.1109/CACSD.2010.5612676}.

\bibitem[Ross(2020)]{rossEnhancementsDIDOOptimal2020}
I.~M. Ross.
\newblock Enhancements to the {{DIDO Optimal Control Toolbox}}, June 2020.
\newblock Comment: 17 pages, 16 figures. New figs explain cotangent barrier and tunneling. Additional references.

\bibitem[Benson et~al.(2006)Benson, Huntington, Thorvaldsen, and Rao]{bensonDirectTrajectoryOptimization2006a}
David~A. Benson, Geoffrey~T. Huntington, Tom~P. Thorvaldsen, and Anil~V. Rao.
\newblock Direct trajectory optimization and costate estimation via an orthogonal collocation method.
\newblock \emph{J. Guid. Control Dyn.}, 29\penalty0 (6):\penalty0 1435--1440, November 2006.
\newblock ISSN 0731-5090, 1533-3884.
\newblock \doi{10.2514/1.20478}.

\bibitem[Garg et~al.(2011{\natexlab{a}})Garg, Hager, and Rao]{gargPseudospectralMethodsSolving2011}
Divya Garg, William~W. Hager, and Anil~V. Rao.
\newblock Pseudospectral methods for solving infinite-horizon optimal control problems.
\newblock \emph{Automatica}, 47\penalty0 (4):\penalty0 829--837, April 2011{\natexlab{a}}.
\newblock ISSN 0005-1098.
\newblock \doi{10.1016/j.automatica.2011.01.085}.

\bibitem[Garg et~al.(2011{\natexlab{b}})Garg, Patterson, Francolin, Darby, Huntington, Hager, and Rao]{gargDirectTrajectoryOptimization2011a}
Divya Garg, Michael~A. Patterson, Camila Francolin, Christopher~L. Darby, Geoffrey~T. Huntington, William~W. Hager, and Anil~V. Rao.
\newblock Direct trajectory optimization and costate estimation of~finite-horizon and infinite-horizon optimal control problems using a {{Radau}} pseudospectral method.
\newblock \emph{Comput Optim Appl}, 49\penalty0 (2):\penalty0 335--358, June 2011{\natexlab{b}}.
\newblock ISSN 1573-2894.
\newblock \doi{10.1007/s10589-009-9291-0}.

\bibitem[Elnagar et~al.(1995)Elnagar, Kazemi, and Razzaghi]{elnagarPseudospectralLegendreMethod1995}
G.~Elnagar, M.A. Kazemi, and M.~Razzaghi.
\newblock The pseudospectral {{Legendre}} method for discretizing optimal control problems.
\newblock \emph{IEEE Transactions on Automatic Control}, 40\penalty0 (10):\penalty0 1793--1796, October 1995.
\newblock ISSN 1558-2523.
\newblock \doi{10.1109/9.467672}.

\bibitem[Fahroo and Ross(2001)]{fahrooCostateEstimationLegendre2001}
Fariba Fahroo and I.~Michael Ross.
\newblock Costate {{Estimation}} by a {{Legendre Pseudospectral Method}}.
\newblock \emph{Journal of Guidance, Control, and Dynamics}, 24\penalty0 (2):\penalty0 270--277, March 2001.
\newblock ISSN 0731-5090.
\newblock \doi{10.2514/2.4709}.

\bibitem[Garg et~al.(2010{\natexlab{a}})Garg, Patterson, Hager, Rao, Benson, and Huntington]{gargUnifiedFrameworkNumerical2010a}
Divya Garg, Michael Patterson, William~W. Hager, Anil~V. Rao, David~A. Benson, and Geoffrey~T. Huntington.
\newblock A unified framework for the numerical solution of optimal control problems using pseudospectral methods.
\newblock \emph{Automatica}, 46\penalty0 (11):\penalty0 1843--1851, November 2010{\natexlab{a}}.
\newblock ISSN 0005-1098.
\newblock \doi{10.1016/j.automatica.2010.06.048}.

\bibitem[Ross and Karpenko(2012)]{rossReviewPseudospectralOptimal2012a}
I.~Michael Ross and Mark Karpenko.
\newblock A review of pseudospectral optimal control: {{From}} theory to flight.
\newblock \emph{Annual Reviews in Control}, 36\penalty0 (2):\penalty0 182--197, December 2012.
\newblock ISSN 1367-5788.
\newblock \doi{10.1016/j.arcontrol.2012.09.002}.

\bibitem[Fran{\c c}olin et~al.(2015)Fran{\c c}olin, Benson, Hager, and Rao]{francolinCostateApproximationOptimal2015}
Camila~C. Fran{\c c}olin, David~A. Benson, William~W. Hager, and Anil~V. Rao.
\newblock Costate approximation in optimal control using integral {{Gaussian}} quadrature orthogonal collocation methods.
\newblock \emph{Optimal Control Applications and Methods}, 36\penalty0 (4):\penalty0 381--397, 2015.
\newblock ISSN 1099-1514.
\newblock \doi{10.1002/oca.2112}.

\bibitem[Garg et~al.(2010{\natexlab{b}})Garg, Patterson, Hager, Rao, Benson, and Huntington]{gargOverviewThreePseudospectral2010}
Divya Garg, Michael~A. Patterson, William~W. Hager, Anil~V. Rao, David~A. Benson, and Geoffrey~T. Huntington.
\newblock An {{Overview}} of {{Three Pseudospectral Methods}} for the {{Numerical Solution}} of {{Optimal Control Problems}}.
\newblock In A.~V. Rao, T.~A. Lovell, F.~K. Chan, and L.~A. Cangahuala, editors, \emph{{{ASTRODYNAMICS}} 2009, {{VOL}} 135, {{PTS}} 1-3}, volume 135, pages 475--+, San Diego, 2010{\natexlab{b}}. Univelt Inc.
\newblock ISBN 978-0-87703-557-2.

\bibitem[Gong et~al.(2008)Gong, Ross, Kang, and Fahroo]{gongConnectionsCovectorMapping2008}
Qi~Gong, I.~Michael Ross, Wei Kang, and Fariba Fahroo.
\newblock Connections between the covector mapping theorem and convergence of pseudospectral methods for optimal control.
\newblock \emph{Comput Optim Appl}, 41\penalty0 (3):\penalty0 307--335, December 2008.
\newblock ISSN 1573-2894.
\newblock \doi{10.1007/s10589-007-9102-4}.

\bibitem[Fahroo and Ross(2006)]{fahrooDiscreteTimeOptimalityConditions2006}
Fariba Fahroo and I.~Michael Ross.
\newblock On {{Discrete-Time Optimality Conditions}} for {{Pseudospectral Methods}}.
\newblock In \emph{{{AIAA}}/{{AAS Astrodynamics Specialist Conference}} and {{Exhibit}}}, Guidance, {{Navigation}}, and {{Control}} and {{Co-located Conferences}}. {American Institute of Aeronautics and Astronautics}, August 2006.
\newblock \doi{10.2514/6.2006-6304}.

\bibitem[{Ernst Hairer} et~al.(2006){Ernst Hairer}, {Gerhard Wanner}, and {Christian Lubich}]{ernsthairerGeometricNumericalIntegration2006}
{Ernst Hairer}, {Gerhard Wanner}, and {Christian Lubich}.
\newblock \emph{Geometric {{Numerical Integration}}}, volume~31 of \emph{Springer {{Series}} in {{Computational Mathematics}}}.
\newblock Springer-Verlag, Berlin/Heidelberg, 2006.
\newblock ISBN 978-3-540-30663-4.
\newblock \doi{10.1007/3-540-30666-8}.

\bibitem[Arribas et~al.(2015)Arribas, Rivo, and Arnedo]{arribasOptimizationPathConstrainedSystems2015}
Daniel~Gonzalez Arribas, Manuel~Sanjurjo Rivo, and Manuel~Soler Arnedo.
\newblock Optimization of {{Path-Constrained Systems}} using {{Pseudospectral Methods}} applied to {{Aircraft Trajectory Planning}}.
\newblock \emph{IFAC-PapersOnLine}, 48\penalty0 (9):\penalty0 192--197, January 2015.
\newblock ISSN 2405-8963.
\newblock \doi{10.1016/j.ifacol.2015.08.082}.

\bibitem[Darby et~al.(2011{\natexlab{a}})Darby, Hager, and Rao]{doi:10.2514/1.52136}
Christopher~L. Darby, William~W. Hager, and Anil~V. Rao.
\newblock Direct trajectory optimization using a variable low-order adaptive pseudospectral method.
\newblock \emph{Journal of Spacecraft and Rockets}, 48\penalty0 (3):\penalty0 433--445, 2011{\natexlab{a}}.
\newblock \doi{10.2514/1.52136}.

\bibitem[Darby et~al.(2011{\natexlab{b}})Darby, Hager, and Rao]{darbyHpadaptivePseudospectralMethod2011}
Christopher~L. Darby, William~W. Hager, and Anil~V. Rao.
\newblock An hp-adaptive pseudospectral method for solving optimal control problems.
\newblock \emph{Optimal Control Applications and Methods}, 32\penalty0 (4):\penalty0 476--502, 2011{\natexlab{b}}.
\newblock ISSN 1099-1514.
\newblock \doi{10.1002/oca.957}.

\end{thebibliography}

\end{document}